\documentclass[11pt]{article}
\usepackage{rheaj}
\usepackage{setspace}

\usepackage{soul}
\usepackage{xparse}

\newcommand{\opt}{{\sc OPT}}
\newcommand{\mypara}[1]{\medskip \noindent {\bf #1}}

\begin{document}
\title{Augmentation based Approximation Algorithms \\ for Flexible Network Design}
\author{
     Chandra Chekuri\thanks{Dept.\ of Computer Science, Univ.\ of Illinois,
    Urbana-Champaign, Urbana, IL 61801. \texttt{chekuri@illinois.edu}. Supported in part by NSF grants
     CCF-1910149 and CCF-1907937.}
   \and
  Rhea Jain\thanks{Dept.\ of Computer Science, Univ.\ of Illinois,
    Urbana-Champaign, Urbana, IL 61801. \texttt{rheaj3@illinois.edu}. Supported in part by NSF grant
    CCF-1910149.}}

\date{\today}

\maketitle

\begin{abstract}
Adjiashvili \cite{Adjiashvili13} introduced network design in a non-uniform fault model: the edge set of a given graph is partitioned into \emph{safe} and \emph{unsafe} edges. A vertex pair $(s,t)$ is $(p,q)$-flex-connected if $s$ and $t$ have $p$ edge-connectivity even after the removal/failure of any $q$ unsafe edges. The goal is to choose a min-cost subgraph $H$ of a given graph $G$ such that $H$ has desired flex-connectivity for a given set of vertex pairs. This model generalizes the well-studied edge-connectivity based network design, however, even special cases are provably much harder to approximate \cite{AdjiashviliHM22,AdjiashviliHMS20}.

The approximability of network design in this model has been mainly studied for two settings of interest: (i) single pair setting under the names FTP and FTF (fault tolerant path and fault tolerant flow) \cite{AdjiashviliHMS20} (ii) spanning setting under the name FGC (flexible graph connectivity) \cite{AdjiashviliHM20,AdjiashviliHM22,BoydCHI22}. There have been several positive results in these papers. However, despite similarity to the well-known network design problems, this new model has been challenging to design approximation algorithms for, especially when $p,q \ge 2$. We obtain two results that advance our understanding of algorithm design in this model.
\begin{itemize}
  \item We obtain a $5$-approximation for the $(2,2)$-flex-connectivity for a single pair $(s,t)$. Previously no non-trivial approximation was known for this setting.
  \item We obtain $O(p)$ approximation for $(p,2)$ and $(p,3)$-FGC for any $p \ge 1$, and for $(p,4)$-FGC for any even $p$. We obtain an $O(q)$-approximation for $(2,q)$-FGC for any $q \ge 1$. Previously only a  $O(q \log n)$-approximation was known for these settings \cite{BoydCHI22}. 
\end{itemize}
Our results are obtained via the augmentation framework where we identify a structured way to use the well-known $2$-approximation for covering uncrossable families of cuts. Our analysis also proves corresponding integrality gap bounds on an LP relaxation that we formulate.
\end{abstract}

\newpage

\section{Introduction}
Network design is an important area of research in discrete and
combinatorial optimization that is motivated by both practical and
theoretical considerations. A broad subclass involves finding a
minimum cost subgraph $H$ of a given graph $G$ that satisfies some
desired connectivity requirements\footnote{We are mainly interested in
  edge-induced subgraphs. Thus, if $G=(V,E)$, $H=(V,F)$ for some
  $F \subseteq E$.}. Special cases include well-studied classical
problems such as the shortest $s$-$t$ path problem, minimum spanning
tree, Steiner tree, and Steiner forest. Generalizations to higher
connectivity requirements such as $k$-edge-connected spanning subgraph
($k$-ECSS) and the survivable network design (EC-SNDP) are important
problems in applications to fault-tolerant network design. Many of
these problems, including Steiner tree, are NP-Hard and APX-hard to
approximate. EC-SNDP (also referred to as the Steiner network problem)
is the following: given connectivity requirements specified by
$r: V \times V \rightarrow \mathbb{Z}_+$ for each pair of vertices
$(u,v)$, find a min-cost subgraph $H$ of $G$ such that each pair
$(u,v)$ is $r(u,v)$-edge-connected in $H$.  EC-SNDP, which captures
several connectivity problems as special cases, admits a
$2$-approximation \cite{Jain01}. Network design problems have been instrumental in the
development of several fundamental and advanced techniques in the
design of approximation algorithms --- see
\cite{Vazirani-approx,WilliamsonS-approx,LauRS-book,KortsarzN10,GuptaK11}.

Over the years, several models have been proposed to capture
robustness of a network to faults. In this paper we are interested in
a specific model that was first proposed by Adjiashvili
\cite{Adjiashvili13}. Recently there has been important algorithmic
progress
\cite{AdjiashviliHM20,AdjiashviliHM22,AdjiashviliHMS20,BoydCHI22} and
we elaborate on the results in these papers after describing the model
and relevant technical background. In this model the input is an edge-weighted
undirected graph $G=(V,E)$ where the edge set $E$ is partitioned to
\emph{safe} edges $\calS$ and \emph{unsafe} edges $\calU$. The
assumption, as the names suggest, is that unsafe edges can fail while
safe edges cannot fail. We say that a vertex-pair $(s,t)$ is
$(p,q)$-flex-connected\footnote{We borrow the terminology from the
  term \emph{flexible graph connectivity} coined in
  \cite{AdjiashviliHM22} and used also in \cite{BoydCHI22}.}  in a
subgraph $H$ of $G$ if $s$ and $t$ are $p$-edge-connected after
deleting from $H$ any subset of at most $q$ unsafe edges. Equivalently,
we require that any cut $\delta(S)$ that separates $s$ from $t$
contains $p$ safe edges or $(p+q)$ edges in total. Network design in
this fault model takes the following form: given $G=(V,E)$ with
$E = \calS \uplus \calU$ and edge costs
$c: E \rightarrow \mathbb{R}_+$, find a min-cost subgraph $H$ such
that $H$ satisfies flex-connectivity given by specification for some
vertex pairs. We observe that this model generalizes standard
edge-connectivity problems. If all edges are safe, that is
$E = \calS$, then asking for $(p,0)$-flex-connectivity for vertex pair
$(s,t)$ is same as asking for $p$-edge-connectivity between $s$ and
$t$. One can model edge-connectivity in an alternate manner; if all
edges are unsafe, that is $E = \calU$, then asking for
$(1,q-1)$-flex-connectivity for $(s,t)$ is same as asking for
$q$-edge-connectivity.  However, $(p,q)$-flex-connectivity is more
general and complex.

In analogy with EC-SNDP, we define the Flex-SNDP problem: the input is
the graph as above with safe/unsafe partition of the edge set, and a
$(p_{u,v},q_{u,v})$-flex-connectivity requirement for each pair
$(u,v)$ of vertices in $G$.  The goal is to find a min-cost subgraph
$H$ of $G$ such that each $(u,v)$ is
$(p_{u,v},q_{u,v})$-flex-connected in $H$. We let $(p,q)$-Flex-SNDP
denote the special case when the requirement for each $(u,v)$ is
either $(p,q)$ or $(0,0)$.  Two special cases will be of main concern
in this paper.  The first is the spanning case which
requires $(p,q)$-flex-connectivity for all pairs of vertices. We refer
to this problem as $(p,q)$-FGC to be consistent with the terminology
introduced by  Boyd et al.\ in \cite{BoydCHI22}. They generalized the
FGC problem (corresponding to $p=1,q=1$) introduced in
\cite{AdjiashviliHM20,AdjiashviliHM22}.  The other special case is
when the requirement is for a single pair $(s,t)$ --- this was the
original motivation for the model in \cite{Adjiashvili13}. We use the
term $(p,q)$-Flex-ST to denote the single pair problem. Previously
$(1,k)$-Flex-ST was referred to as FTP (Fault-Tolerant Path) in
\cite{Adjiashvili13,AdjiashviliHMS20} and the $(k,1)$-Flex-ST problem
was referred to as FTF (Fault-Tolerant Flow) in \cite{AdjiashviliHMS20}; these problems were
studied in directed graphs but it was noted in \cite{AdjiashviliHMS20}
that the undirected FTP reduces to directed FTP. It is also not hard
to see that undirected FTF can be reduced to directed FTF at a slight
loss in the approximation.  We use the term $(p,q)$-Flex-Steiner to
refer to an instance where there is a $(p,q)$ requirement for every
pair of vertices from some given set $T \subseteq V$ of terminal
vertices. Note that when $T=V$ we have $(p,q)$-FGC and when
$T = \{s,t\}$ we have the $(p,q)$-Flex-ST problem.

\mypara{Motivation and recent work:} The model of Adjiashvili is
motivated by the desire to incorporate non-uniformity in robust
network design. This aims to bridge the gap between the standard connectivity problems
that are clean and tractable and the models which
are more general but tend to be less tractable --- we refer
the reader to \cite{Adjiashvili13,AdjiashviliHMS20,AdjiashviliHM22} for further discussion.

Apart from the practical motivation, we find the theoretical aspects
of the model compelling for several reasons that we outline below.
Adjiashvili, Hommelsheim and M\"uhlenthaler \cite{AdjiashviliHM20}
introduced FGC, pointed out that it generalizes the well-known MST
and 2-ECSS problems, and derived a constant factor approximation.
Boyd et al. \cite{BoydCHI22} obtained several results on
$(p,q)$-FGC. They improved the approximation ratio for FGC to
$2$. They showed a $\min\{(q+1),O(\log n)\}$-approximation for $(1,q)$-FGC and a
$4$-approximation for $(p,1)$-FGC. They obtained an
$O(q \log n)$-approximation for $(p,q)$-FGC for any $q$.  Importantly,
they show several strong connections between flexible graph
connectivity and \emph{capacitated} network design which has been
studied in several works \cite{GoemansGPSTW94,CarrFLP00,ChakCKK15,ChakKLN13}. Capacitated network design
generalizes the standard edge connectivity setting by allowing each edge
$e$ to have an integer capacity $u_e \ge 1$. One can reduce
capacitated network design to standard edge-connectivity network
design by replacing each edge $e$ by $u_e$ parallel edges with
capacity $1$ each. This blows up the approximation factor by
$\max_e u_e$, which is acceptable when this quantity is small. When
$\max_e u_e$ can be large, the complexity of capacitated network
design varies.
While the single pair problem becomes hard to
approximate to almost polynomial factors \cite{ChakKLN13}, the spanning case
(when one seeks to find a min-cost subgraph that has connectivity at
least a given quantity $R$) admits an $O(\log n)$-approximation
\cite{ChakCKK15}. Boyd et al. show that $(1,k)$-Flex-SNDP and $(k,1)$-Flex-SNDP
can be reduced to Cap-SNDP such that the maximum capacity is $k$;
this implies that $(1,k)$-Flex-SNDP and $(k,1)$-Flex-SNDP admit an
$O(k)$ approximation.  Although these ratios are not necessarily tight
in all cases, they provide a first-order and easy approach to solve
the $(1,k)$ and $(k,1)$ cases. Boyd et al.\ also show that an
important technique in network design, namely the augmentation
approach based on covering uncrossable families, can be applied in
some cases; they obtain a $4$-approximation for $(k,1)$-FGC via this approach.

Adjiashvili et al.\ \cite{AdjiashviliHMS20} considered the single pair
setting, mainly in directed graphs: specifically $(1,k)$-Flex-ST and
$(k,1)$-Flex-ST problems (referred to as FTP and FTF in
\cite{AdjiashviliHMS20}). They obtained a $k$ approximation for
$(1,k)$-Flex-ST, slightly improving the $(k+1)$-approximation that can
obtained via an LP relaxation; they prove that the undirected graph
case reduces to the directed case. They prove that $(1,k)$-Flex-ST in
directed graphs is at least as hard as directed Steiner tree which
implies poly-logarithmic factor inapproximability
\cite{HalperinK03}. For $(k,1)$-Flex-ST they obtain a
$(k+1)$-approximation; when $k$ is a fixed constant, they obtain a
$2$-approximation via an involved dynamic programming based approach.
They prove that $(k,1)$-Flex-ST in directed and undirected graphs is
at least as hard to approximate as directed Steiner forest (which has
almost polynomial factor hardness \cite{DodisK99}).  The hardness
results are when $k$ is part of the input and large, and show that
approximability of network design in this model is substantially
different from the edge-connectivity model.

The results discussed so far, the practical motivation, and the fact
that a natural LP relaxation requires $n^{O(q)}$-time to solve,
suggest that it is fruitful to focus on the approximability of
$(p,q)$-flex-connectivity network design when $p,q$ are small
constants. It is natural to conjecture that $(p,q)$-Flex-SNDP admits
an $f(p,q)$-approximation for a non-negative integer valued function
$f: \mathbb{Z}_+ \times \mathbb{Z}_+ \rightarrow \mathbb{Z}_+$.  Note
that for $(1,k)$ and $(k,1)$ cases this is true. A
weaker version of the conjecture is that $(p,q)$-Flex-SNDP admits an
$f(p,q)\text{polylog}(n)$-approximation.  For $(p,q)$-FGC we already
had an $O(q \log n)$-approximation and it is natural to conjecture an
$O(q)$-approximation, or even a constant factor approximation. As far
as we know there are no hardness results or lower bounds on the LP
integrality gap that rule out these conjectures. From a theoretical
point of view the model presents novel and interesting challenges for
algorithm design.

\subsection{Our contribution}
Despite the possibility of an $f(p,q)$-approximation for
$(p,q)$-Flex-SNDP we did not have any non-trivial approximation for
$(2,2)$-Flex-ST problem! We note that $(2,2)$-FGC admits an
$O(\log n)$-approximation.  Thus $(2,2)$-Flex-ST is the first
interesting case where both $p, q > 1$ and is not a spanning
problem.  Our first result is the following.

\begin{theorem}
  \label{thm:intro-st}
  There is a $5$-approximation for $(2,2)$-Flex-ST problem.
\end{theorem}

A natural question is whether one can obtain a non-trivial
approximation for $(2,2)$-Flex-Steiner problem. We do not yet have a
direct approach for it. The preceding result easily implies an
$O(h)$-approximation where $h$ is the number of terminals.
Based on past work on vertex connectivity network
design \cite{ChakrabortyCK08,ChuzhoyK08,ChekuriK08,Korula-thesis}, we propose a simple and natural algorithm that we conjecture
yields an $O(\log |h|)$-approximation for $(2,2)$-Flex-Steiner
problem. We discuss this in Section~\ref{sec:steiner}.

Our second set of results are for $(p,q)$-FGC. Recall that Boyd et
al.\ \cite{BoydCHI22} obtained a $4$-approximation for $(k,1)$-FGC,
a $\min\{(k+1), O(\log n)\}$-approximation for $(1,k)$-FGC, and a
$O(q \log n)$-approximation for $(p,q)$-FGC. Thus we did not have a
constant factor approximation for $(2,2)$-FGC.  We prove several
results that, as a corollary, yield constant factor approximation for
small values of $p,q$.

\begin{theorem}
\label{thm:introfgc}
For any $q \ge 0$ there is a $(2q+2)$-approximation for $(2,q)$-FGC.
For any $p \ge 1$  there is a $(2p+4)$-approximation for $(p,2)$-FGC,  and a
$(4p+4)$-approximation for $(p,3)$-FGC.  Moreover, for
all \emph{even} $p \ge 2$ there is an $(6p+4)$-approximation for $(p,4)$-FGC.
\end{theorem}

We explicitly formulate an LP relaxation for $(p,q)$-Flex-Steiner
combining ideas implicit in
\cite{AdjiashviliHM22,AdjiashviliHMS20,BoydCHI22} and capacitated
network design where knapsack cover inequalities play an important
role \cite{CarrFLP00,ChakCKK15}. Although we do not use the LP
relaxation directly, our approximation bounds can be shown with
respect to this relaxation.  Proving lower bounds on the integrality
gap of this relaxation for various special cases of $(p,q)$-Flex-SNDP
could yield insights into the hardness of the problem. For instance we
show that the LP gap is $\Omega(k)$ for $(1,k)$-Flex-ST (see Section~\ref{sec:integrality_gap_1k}).
We provide examples showing the need for new ideas to extend our
current results for $(p,q)$-FGC.

\mypara{Barriers and techniques:} The non-uniform nature of flexible
connectivity model makes it technically challenging.  The reduction to
capacitated network design does not extend when $p, q \ge 2$.
Although we can formulate an LP relaxation, it does not have clean
structural properties that edge-connectivity network design
enjoys. Thus, known techniques such as primal-dual and iterated
rounding cannot be applied directly. Another technique in network
design is the \emph{augmentation} approach which has proven to be
quite successful --- in fact, this is the approach that was used for
EC-SNDP prior to the iterated rounding approach.  Moreover, it is
one of the main techniques that we have in more complex settings such
as vertex connectivity network design \cite{KortsarzN05,Nutov12} and node-weighted network design \cite{Nutov10,ChekuriEV12}. In the
augmentation approach we start with a initial set of edges $F_0$ that
partially satisfy the connectivity constraints. We then augment $F_0$
with a set $F$ in the graph $G-F_0$; the augmentation is typically
done to increase the connectivity by one unit for pairs that are not
yet satisfied. We repeat this process in several stages until all
connectivity requirements are met. A crucial advantage of the
augmentation approach is that it allows one to reduce a
higher-connectivity problem to a series of problems that solve
a potentially simpler $\{0,1\}$-connectivity problem.  An important tool in this
area is a $2$-approximation for covering an \emph{uncrossable}
function (a formal definition is given in Section~\ref{sec:prelim})
\cite{WilliamsonGMV95}.

Boyd et al.\ \cite{BoydCHI22} used the augmentation approach to obtain
a $4$-approximation for $(p,1)$-FGC. They first obtain a
$2$-approximate solution to the $(p,0)$ problem, which is the same as
the $p$-ECSS problem. In the second step, they augment the solution to
obtain a feasible solution to $(p,1)$-FGC. An advantage of augmenting
from $(p,q)$-FGC to $(p,q+1)$-FGC is that one can ignore the
distinction between safe and unsafe edges in the augmentation step.
\cite{BoydCHI22} shows that the augmentation problem from $(p,0)$ to
$(p,1)$ results in a nice problem: the cuts to be covered are
uncrossable.  One can try this approach for $(p,q)$-Flex-SNDP wherein we
incrementally augment from $(p,q-1)$ to $(p,q)$ by defining an
appropriate set of cuts to cover. It is not too hard to show that the
augmentation problem from $(2,q-1)$-FGC to $(2,q)$-FGC is uncrossable
for any $q$.  However, the main difficulty is that for most values of
$p,q$, the augmentation problem does \emph{not} lead to an uncrossable
family of cuts. We overcome this difficulty by decomposing the family
of cuts to be covered in the augmentation problem into a sequence of
cleverly chosen subfamilies such that each subfamily is
uncrossable. The process is very problem specific: our approach for
$(2,2)$-Flex-ST is very different from that for $(p,q)$-FGC although
they share the high-level approach. There is antecedent for such
complexity in vertex-connectivity network design; $\{0,1,2\}$-VC-SNDP
admits a $2$-approximation \cite{FleischerJW06,CheriyanVV06} while
higher connectivity requires different techniques
\cite{ChakrabortyCK08,ChuzhoyK08,ChuzhoyK12,Nutov12}, and
the approximability and integrality  gaps are not fully resolved yet.  We hope that our
results will spur additional insights and future work in flexible
network design.

\mypara{Organization:} The rest of the paper is organized as
follows. Section~\ref{sec:prelim} discussed the LP relaxation and some
of the technical tools needed for the augmentation
framework. Section~\ref{sec:22st} describes the $5$-approximation for
$(2,2)$-Flex-ST. Section~\ref{sec:fgc} contains our results for
$(p,q)$-FGC.

\section{Preliminaries}
\label{sec:prelim}
Throughout the paper we will assume that we are given an undirected graph
$G=(V,E)$. Unless stated otherwise, we
will assume that $E$ is partitioned into safe $\calS$ and unsafe
$\calU$ edges. We are interested in edge-induced subgraphs of a
graph. Thus, when we say that $H$ is a subgraph of $G=(V,E)$ we
implicitly assume that $H=(V,F)$ for some $F \subseteq E$. For any
subset of edges $F \subseteq E$ and any set $S \subseteq V$
we use the notation $\delta_F(S)$ to denote the set of edges in $F$ that
have exactly one end point in $S$.  We may drop $F$ if it is clear from the
context.

\mypara{LP Relaxation:} We describe an LP relaxation for
$(p,q)$-Flex-Steiner problem.  It is straight forward to generalize it for Flex-SNDP but we do not explicitly describe it here.  Recall that we are given set of terminals $T \subseteq V$ and the goal is to
choose a min-cost subset of the edges $F$ such that in the subgraph
$H=(V,F)$, $u$ and $v$ are $(p,q)$-flex-connected for any $u,v \in
T$. Let $\calC = \{S \subset V \mid S \cap T \neq T \}$ be the set of
all vertex sets that separate some terminal pair. For a set of edges $F$ to be feasible for the given $(p,q)$-Flex-Steiner instance, we require that for all $S \in \calC$, $|\delta_F(S) \setminus B| \ge p$ for any $B \subseteq \calU$ with $|B| \le q$. We can write cut covering constraints expressing this condition, but these constraints are not adequate by themselves. To improve this LP, we consider the connection to capacitated network design: we give each safe edge a capacity of $p+q$, each unsafe edge a capacity of $p$, and require $p(p+q)$ connectivity for the terminal pairs; it is not difficult to verify that this is a valid constraint. These two sets of constraints yield the following LP relaxation with variables $x_e \in [0,1]$, $e \in E$. 

\begin{align*}
    \min \sum_{e \in E} c_ex_e &\\
    \text{subject to} \sum_{e \in \delta(S) - B} x_e &\geq p  &S \in \calC, B \subseteq \calU, |B| \le q \\
    (p+q)\sum_{e \in S \cap \calS} x_e + p \sum_{e \in S \cap \calU}
  x_e &\geq p(p+q) & S \in \calC \\
    x_e & \ge [0,1] &e \in E
\end{align*}

It is not hard to see that the preceding LP admits a separation oracle
that runs in $n^{O(q)}$ time; for each $B \subseteq \calU, |B| \le q$
we remove $B$ and check that for each $s,t \in T$, the $s$-$t$ min-cut
value in the graph $G-B$ with edge-capacities given by $x$ is at least
$p$. For the capacitated constraints, we can give every safe edge $e$ a
weight of $(p+q)x_e$, every unsafe edge $e$ a weight of $p(x_e)$, and check that the
$s$-$t$ min-cut value is at least $p(p+q)$ for each $s, t \in
T$. Hence the LP can be solved in $n^{O(q)}$ time.  Our algorithms do
not directly use the preceding LP relaxation, however, all of our
approximation bounds can be shown with respect to the lower bound
provided by it.

For the special case of $(p,q)$-FGC we can show that the LP can be
solved in polynomial time without any dependence on $q$. We borrow
ideas from \cite{BoydCHI22,ChakCKK15}. We defer the proof to the appendix (Section~\ref{sec:appendix_prelim_proofs}).

\begin{lemma}
\label{lemma:lp_polytime_fgc}
  The LP relaxation can be solved in polynomial time for $(p, q)$-FGC.
\end{lemma}

\begin{remark}
  Boyd et al.\ \cite{BoydCHI22} obtained several results for
  $(p,q)$-FGC. Several of them can be shown with
  respect to the lower bound provided by the LP relaxation above. We
  defer details to a future version. One can also write a similar LP
  relaxation for the directed version of problems.  In some settings
  it may be advantageous to use the directed formulation even for
  undirected problems.
\end{remark}

\mypara{Augmentation:} The results of this paper rely on the
augmentation framework. Suppose $G=(V,E),T \subseteq V$ is an
instance of $(p,q)$-Flex-Steiner.
We observe that  $(p,0)$-Flex-Steiner instance can be solved via
$2$-approximation to EC-SNDP hence we are interested in $q \ge 1$.
Let $F_1$ be a feasible solution for the $(p,q-1)$-Flex-Steiner instance on $G,T$.
This implies that for any cut $S$ that separates two terminals we have
$|\delta_{F_1 \cap \calS}(S)| \ge p$ or $|\delta_{F_1}(S)| \ge p+q-1$.
We would like to augment $F_1$ to obtain a feasible solution to
satisfy the $(p,q)$ requirement. Define a function $f: 2^{|V|} \to \{0, 1\}$ where
$f(S) = 1$ iff (i) $S$ separates terminals and (ii) $|\delta_{F_1 \cap
  \calS}(S)| < p$ \emph{and} $|\delta_{F_1}(S)| = p+q-1$.  We call $S$ a \emph{violated} cut with respect to $F_1$. Since
$F_1$ satisfies $(p, q-1)$ requirement, if $|\delta_{F_1
  \cap \calS}| < p$ it must be the case that $|\delta_{F_1}(S)| \ge p+q-1$. The following lemma is simple.

\begin{lemma}
  Suppose $F_2 \subseteq E \setminus F_1$ is a feasible cover for $f$, that is,
  $\delta_{F_2}(S) \ge f(S)$ for all $S$. Then $F_1 \cup F_2$ is a feasible solution
  to $(p,q)$-flex-Steiner for the terminal set $T$.
\end{lemma}

The augmentation problem is then to find a min-cost subset of edges to
cover $f$ in $G-F_1$.  The key observation is that the augmentation
problem does not distinguish between safe and unsafe edges and hence
we can rely on traditional connectivity augmentation ideas.  Note that
$f$ is a $\{0,1\}$ function.  We prove a stronger lemma, showing that
the LP relaxation for the original instance provides a valid
cut-covering relaxation for the augmentation problem. The proof is deferred to the appendix (Section~\ref{sec:appendix_prelim_proofs}).

\begin{lemma}
\label{lemma:augmentation_lp}
  Let $x \in [0,1]^{|E|}$ be a feasible LP solution for a given
  instance of $(p,q)$-Flex-Steiner. Let $F_1$ be a feasible solution
  that satisfies $(p,q-1)$ requirements for the terminal.  Then, for
  any violated cut $S \subseteq V$ in $(V, F_1)$, we have
  $\sum_{e \in \delta(S) \setminus F_1} x_e \geq 1$.
\end{lemma}

One can try to augment from $(p,q)$ to
$(p+1,q)$-flex-connectivity. However, the resulting augmentation
problem is less well-behaved; we still need to distinguish between
safe and unsafe edges.

\mypara{Uncrossable functions and families:} Uncrossable functions are a general class
of requirement functions that are an important ingredient in network design \cite{WilliamsonGMV95,GoemansW97,GuptaK11,KortsarzN10}.

\begin{definition}
A function $f: 2^{|V|} \to \{0, 1\}$ is \emph{uncrossable} if for
every $A, B \subseteq V$ such that $f(A) = f(B) = 1$, one of the
following is true: (i) $f(A \cup B) = f(A \cap B) = 1$, (ii)  $f(A - B) = f(B - A) = 1$.
A family of cuts $\calC \subset 2^V$ is an uncrossable \emph{family} if the indicator function
$f_{\calC}:2^V \rightarrow \{0,1\}$ with $f(S) = 1$ iff $S \in \calC$, is uncrossable.
\end{definition}

\begin{definition}
  Let $G=(V,E)$ be a graph and let
  $f: 2^{|V|} \to \{0, 1\}$ be a requirement function. Let $A \subseteq E$. A set
  $S$ is violated with respect to $A,f$ if $f(S) = 1$ and $\delta_A(S) = \emptyset$.
  The residual requirement function of $f$ with respect to $A$, denoted by $f_A$ is
  indicator function of the violated sets of $f$ with respect to $A$.
\end{definition}

An important result in network design is a $2$-approximation for the problem of
min-cost covering of an uncrossable function $f$ by the edge set of a graph \cite{WilliamsonGMV95}; this was
later generalized to covering all skew-supermodular requirement functions \cite{Jain01}.
This requires a computational assumption on $f$ and we encapsulate it below.

\begin{theorem}[\cite{WilliamsonGMV95}]
  \label{thm:uncrossable}
  Let $G=(V,E)$ be an edge-weighted graph and let $f:2^V \rightarrow \{0,1\}$ be an uncrossable
  function. Suppose there is an efficient oracle that for any $A \subseteq E$ outputs all the
  minimal violated sets of $f$ with respect to $A$. Then there is an efficient $2$-approximation for the
  problem of finding a minimum cost subset of edges that covers $f$.
\end{theorem}

A special case of uncrossable family of sets is a \emph{ring family}.
We say that an uncrossable family of sets $\calC \subseteq 2^V$ is a
ring family if the following conditions hold: (i) if $A,B \in \calC$
and $A,B$ properly intersect\footnote{$A, B$ properly intersect if $A \cap B \neq \emptyset$ and $A-B, B-A \neq \emptyset$.} then $A \cap B$ and $A \cup B$ are in
$\calC$ and (ii) there is a unique minimal set in $\calC$.  We observe
that if we have an uncrossable family $\calC$ such that there is a
vertex $s$ contained in every $A \in \calC$ then $\calC$ is
automatically a ring family.  Theorem~\ref{thm:uncrossable} can be
strengthened for this case. There is an optimum algorithm to find a min-cost
cover of a ring family --- see \cite{Nutov10,Nutov12,Frank1979}

\paragraph{Enumerating small cuts in a graph:} In order to use
Theorem~\ref{thm:uncrossable} in the augmentation framework, we need to
able to first prove that a family of cuts that we intend to cover is
uncrossable.  Second, we need to be able efficiently find all the
minimal violated sets of the family with respect to a given set of
edges $A$. Consider the augmentation problem from $(p,q)$ to $(p,q+1)$,
and the requirement function $f$ that is induced by it. Recall that
$f(S) = 1$ iff $|\delta_{F_1}(S)| = p+q$ and
$|\delta_{F_1 \cap \calS}| < p$. Thus, for any fixed $p,q$ we can
enumerate all violated sets in $n^{O(p+q)}$ time by trying all
possible cuts in $F_1$ with $p+q$ edges.  Consider the graph
$G'=G-F_1$ and any set of edges $A \subseteq E \setminus F_1$. It
follows that given $A$ and $F_1$ that is feasible for $(p,q)$, we can
find the set of \emph{all} violated cuts in $G-F_1$ with respect to
$A,f$ in $n^{O(p+q)}$ time. In the context of $(p,q)$-FGC we need a more sophisticated process since
we do not limit ourselves to fixed $p,q$. The following lemma from \cite{BoydCHI22} is useful
for this.

\begin{lemma}[\cite{BoydCHI22}]
  Let $G=(V,E)$ be an instance of $(p,q)$-FGC where $p,q \ge 1$. Let $F_1 \subseteq E$ be a
  feasible solution to $(p,q-1)$-FGC on $G$. Let $\calC$ be the set of all cuts that
  are violated with respect to the augmentation function $f$. Then, there are $O(n^{4})$
  such cuts and they can  be enumerated in polynomial time.
\end{lemma}

\mypara{Submodularity and posimodularity of the cut function:} It is
well-known that the cut function of an undirected graph is symmetric
and submodular.  Submodularity implies that for all
$A, B, \subseteq V$, $|\delta(A)| + |\delta(B)| \ge |\delta(A \cap B)| + |\delta(A \cup B)|$.
Symmetry and submodularity also implies posimodularity: for all $A, B
\subseteq V$, $|\delta(A)| + |\delta(B)| \ge |\delta(A - B)| + |\delta(B - A)|$.

\section{A $5$-approximation for $(2, 2)$-Flex-ST}
\label{sec:22st}
In this section, we provide a constant factor approximation for
$(2, 2)$-Flex-ST. We are given a graph $G = (V, E)$ with a cost
function on the edges $c: E \to \R_{\geq 0}$, a partition of the edge
set $E = \calS \uplus \calU$ into safe and unsafe edges, and
$s, t \in V$. Our goal is to find the cheapest set of edges such that
every $s$-$t$ cut has either at least two safe edges or at least four
total edges. Let $F^* \subseteq E$ denote the optimal solution. Note
that $F^*$ is a feasible solution to the $(2,1)$-Flex-ST instance for
the same graph $G$ and the same pair $s,t$. The reduction to the
Cap-SNDP problem in $\cite{BoydCHI22}$ implies a $3$-approximation
for $(2,1)$-Flex-ST. It is also worth noting that Adjiashvili,
Hommelsheim, Mühlenthaler, and Schaudt \cite{AdjiashviliHMS20}
describe a $2$-approximation to $(k,1)$-Flex-ST in directed graphs for
any fixed $k$, which can be modified to give a 3-approximation in
undirected graphs\footnote{This is not explicitly stated in
  \cite{AdjiashviliHMS20} but follows from an easy observation by
  bidirecting edges. Unlike the case of $(1,k)$-Flex-ST
  \cite{AdjiashviliHMS20}, the directed and undirected graph cases do
  not seem to equivalent for $(k,1)$-Flex-ST.}.

We follow the augmentation approach outlined earlier. However, instead
of finding a feasible solution to $(2,1)$-Flex-ST, we start with
something slightly stronger via the reduction to capacitated network
design.

\begin{lemma}
\label{st_22_starting_criteria}
There exists a set of edges $F \subseteq E$ such that for any $s$-$t$
separating set $A\subset V$, exactly one of the following is true: (1) $|\delta_{F \cap \calS}(A)| \geq 2$, (2) $|\delta_F(S)| \geq 4$, (3) $\delta_{F}(A)$ has exactly two unsafe and one safe
edge. Furthermore, we can efficiently find such a set $F$ such that
$\cost(F) \leq 2 \cdot \cost(F^*)$.
\end{lemma}
\begin{proof}
  Consider giving every safe edge a capacity of $2$, every unsafe edge
  a capacity of $1$, and approximating the solution to the
  corresponding Cap-SNDP problem with an $s$-$t$ connectivity
  requirement of $4$ (all other pairs have $0$ requirement). Suppose
  $F \subseteq E$ is a feasible solution to this Cap-SNDP
  instance. Then, one can easily verify that $F$ follows the above
  requirements by casing on the number of safe edges in $\delta(A)$:
  \begin{enumerate}
  \item if |$\delta_{F \cap \calS}(A)| = 0$, since capacity of
    $\delta_F(A)$ is $4$, it contains at least four unsafe edges.
  \item else if |$\delta_{F \cap \calS}(A)| = 1$, then the safe edge
    provides capacity of $2$ and to reach capacity $4$, $\delta_F(A)$
    contains at least two unsafe edges
  \item else |$\delta_{F \cap \calS}(A)| \ge 2$.
\end{enumerate}
We can approximate the Cap-SNDP instance to within a factor of $2$ by
replacing each edge of capacity $2$ by two edges of capacity $1$ and
with the same cost. This reduces the problem to solving an $s$-$t$
minimum-cost flow problem (flow of value $4$ with all edges having
capacity $1$) which can be solved optimally. Thus we obtain the
desired polynomial time $2$-approximation.

Finally, we show that any solution $F' \subseteq E$ to $(2,2)$-Flex-ST
is also a solution to this Cap-SNDP instance. Let $A \subseteq V$ be
an $s$-$t$ cut. Then, $\delta_{F'}(A)$ has two safe edges or four
total. In both cases, $\delta_{F'}(A)$ has a total capacity of at
least four. Thus the $2$-approximation to the Cap-SNDP instance costs
at most twice the optimum solution to $(2,2)$-Flex-ST.
\end{proof}

Let $F \subseteq E$ be obtained via the algorithm in the preceding
lemma, so for every $s$-$t$ cut $A$, $\delta_F(A)$ has at least two
safe edges or at least four total edges or exactly two unsafe and one
safe edge. If any cut has at least two safe or four total edges, it
satisfies the requirement on $(2,2)$-Flex-ST, so we focus our
attention on those cuts with exactly two unsafe and one safe edge. Let
$\calC$ denote the set of all violated cuts containing $s$,
i.e.
$\calC = \{A \subset V: s \in A, t \notin A, |\delta_{F \cap
  \calS}(A)| = 1, |\delta_{F \cap \calU}(A)| = 2\}$. By symmetry, it
suffices to only consider cuts containing $s$, since covering a set
also covers its complement. Thus this is exactly the family of cuts
that we need to cover in the augmentation phase. Since all
violated cuts have exactly three edges, we can assume without loss of
generality that all remaining edges in $E - F$ are unsafe. Therefore,
it suffices to solve the augmentation problem: we want to find the
minimum cost subset $F' \subseteq E \setminus F$ s.t.
$\delta_{F'}(A) \geq 1$ for all $A \in \calC$. Note that for any
$A \in \calC$, $t \in V\setminus A$ and hence for any
$A, B \in \calC$, $A \cup B \neq V$.  This family of violated cuts is
unfortunately not uncrossable, as shown in
Figure~\ref{st_22_counterex}. Instead, we will show that we can find
three sub-families of $\calC$ whose union is $\calC$ and each
sub-family is a ring family.

\begin{figure}
  \begin{center}
    \includegraphics[width = 0.4\linewidth]{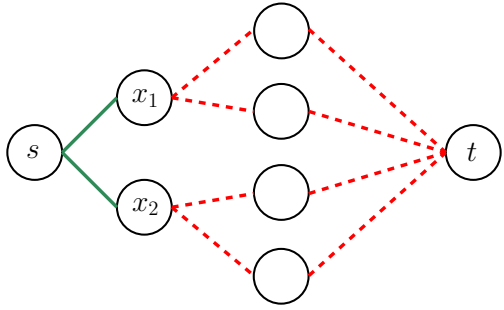}
  \end{center}
  \caption{Example graph where violated cuts are not uncrossable. Red dashed edges are unsafe, while green solid edges are safe. Let
    $A = \{s, x_1\}, B = \{s, x_2\}$. Notice that $A$ and $B$ are
    violated, but $A \cup B = \{s, x_1, x_2\}$ and $A \cap B = \{s\}$
    are not. Also $A-B$ and $B-A$ are not violated, since they don't separate $s$ and $t$.}
  \label{st_22_counterex}
\end{figure}

We begin with a lemma that characterizes when two violated sets $A,B$ do not uncross.

\begin{lemma}
\label{st_22_crossing_criteria}
Suppose $A, B \in \calC$ and the sets properly intersect. Then, either
$A \cup B, A \cap B \in \calC$, or one of $\delta_F(A \cup B)$ and
$\delta_F(A \cap B)$ has exactly two safe edges (and no unsafe) and
the other has exactly four unsafe edges (and no safe). These
structures are demonstrated in Figure~\ref{st_22_crossing_image}.
\end{lemma}

\begin{proof}
  Suppose $A, B \in \calC$ are two sets that properly intersect. Note
  that they must each have three edges crossing them, out of which one
  is safe and two are unsafe. First, suppose neither
  $\delta_F(A \cup B)$ nor $\delta_F(A \cap B)$ contain two safe
  edges. Then, they must have at least three total edges since $F$ is
  feasible for $(2,1)$ $s$-$t$-connectivity. By submodularity,
\[|\delta_F(A)| + |\delta_F(B)| = 6 \geq |\delta_F(A \cup B)| + |\delta_F(A \cap B)|.\]
Therefore, $\delta_F(A \cup B)$ and $\delta_F(A \cap B)$ must both have exactly three total edges. Since neither of the two have two safe edges, they are both violated, i.e. $A \cup B, A \cap B \in \calC$.

Instead, suppose one of $\delta_F(A \cup B)$ and $\delta_F(A \cap B)$
has at least two safe edges; without loss of generality, we assume
$\delta_{F \cap S}(A \cap B) \geq 2$. Then, by submodularity of the
cut function,
\[|\delta_{F \cap \calS}(A)| + |\delta_{F \cap \calS}(B)| = 2 \geq |\delta_{F \cap \calS}(A \cap B)| + |\delta_{F \cap \calS}(A \cup B)|.\]
This tells us that $|\delta_{F \cap S}(A \cap B)| \leq 2$, so $A \cap B$ must be crossed by exactly 2 safe edges. Furthermore, if $|\delta_{F \cap S}(A \cap B)| = 2$, then $|\delta_{F \cap S}(A \cup B)| = 0$. Recall that all cuts in $F$ are crossed by either at least two safe edges, at least four total edges, or exactly one safe and two unsafe edges. Therefore, if $\delta_F(A \cup B)$ has no safe edges, it must be the case that $|\delta_F(A \cup B)| \geq 4$. However, we know from above that $|\delta_F(A \cup B)| + |\delta_F(A \cap B)| \leq 6$. Therefore, $|\delta_F(A \cup B)| = 4$ and $|\delta_F(A \cap B)| = 2$. Since $\delta_F(A \cup B)$ has no safe edges, it must be the case that $\delta_F(A \cup B)$ has exactly four unsafe edges (and no safe). Similarly, since $|\delta_{F \cap S}(A \cap B)| = |\delta_F(A \cap B)| = 2$, $\delta_F(A \cap B)$ has exactly two edges, both of which are safe. The argument is analogous in the case that $\delta_F(A \cup B)$ has two safe edges.
\end{proof}

\begin{figure}
  \begin{center}
    \includegraphics[width = 0.5\linewidth]{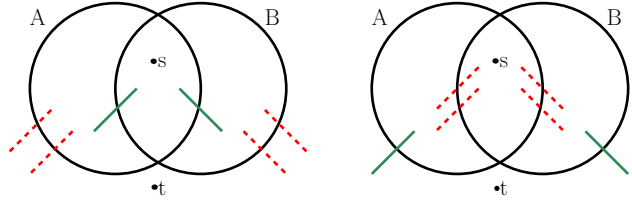}
  \end{center}
  \caption{Crossing Sets for FGC s-t (2, 2) Augmentation. Red dashed edges are unsafe, while green solid edges are safe.}
  \label{st_22_crossing_image}
\end{figure}

\begin{corollary}
\label{st_22_shared_edge}
For each safe edge $e \in \calS$, the set of all violated cuts $A \subseteq V$ with $e \in \delta_F(A)$ is a ring family.
\end{corollary}
\begin{proof}
Let $A, B \in \calC$. Recall that this means they must each be crossed by exactly one safe edge. Suppose $\delta_{F \cap \calS}(A) = \delta_{F \cap \calS}(B) = \{e\}$. Any safe edge crossing $A \cap B$ or $A \cup B$ also crosses at least one of $A$ and $B$, since $\delta_{F \cap \calS}(A \cap B) \subseteq \delta_{F \cap \calS}(A) \cup \delta_{F \cap \calS}(B)$ and $\delta_{F \cap \calS}(A \cup B) \subseteq \delta_{F \cap \calS}(A) \cup \delta_{F \cap \calS}(B)$. Therefore, both $A \cap B$ and $A \cup B$ are crossed by at most 1 safe edge. By Lemma~\ref{st_22_crossing_criteria}, $A \cup B, A \cap B \in \calC$. Because of this, they must both be crossed by exactly one safe edge, so $e \in \delta_{F}(A \cap B)$ and $e \in \delta_F(A \cup B)$.
\end{proof}

\begin{lemma}
\label{st_22_path_uncross}
Suppose $A, B \in \calC$ and the sets properly intersect, and let
$e_1, e_2$ denote the safe edges in $\delta_F(A)$ and $\delta_F(B)$
respectively. If there is an $s$-$t$ path in the graph $(V,F)$
containing both $e_1$ and $e_2$, then $A \cup B, A \cap B \in
\calC$. Furthermore,
$\delta_{F \cap \calS}(A \cap B) \cup \delta_{F \cap \calS}(A \cup B)
\subseteq \{e_1, e_2\}$.
\end{lemma}
\begin{proof}
 Let $A, B \in \calC$ such that the sets properly intersect,
  $\{e_1\} = \delta_{F \cap \calS}(A)$ and
  $\{e_2\} = \delta_{F \cap \calS}(B)$, and $P$ be an $s$-$t$ path in
  the graph $(V,F)$ containing both $e_1$ and $e_2$.   If $e_1 = e_2$, then by Corollary~\ref{st_22_shared_edge}, we are done. Else, without loss of
  generality, suppose $e_1$ comes before $e_2$ on this path
  $P$. Suppose for the sake of contradiction, $A \cap B \notin \calC$
  or $A \cup B \notin \calC$. By Lemma~\ref{st_22_crossing_criteria},
  either $|\delta_{F \cap \calS}(A \cap B)| = 2$ or
  $|\delta_{F \cap \calS}(A \cup B)| = 2$. Without loss of generality,
  suppose $|\delta_{F \cap \calS}(A \cap B)| = 2$. Since
  $\delta_{F \cap \calS}(A \cap B) \subseteq \delta_{F \cap \calS}(A)
  \cup \delta_{F \cap \calS}(B)$,
  $\delta_{F \cap \calS}(A \cap B) = \{e_1, e_2\}$.

  Notice that Lemma~\ref{st_22_crossing_criteria} also tells us that
  there are no other edges in $\delta_F(A \cap B)$. Let
  $e_1 = \{x_1, y_1\}, e_2 = \{x_2, y_2\}$ where $x_i \in A \cap B$,
  $y_i \notin A \cap B$ for $i = 1, 2$. Then, since $s \in A \cap B$
  and $e_1$ comes before $e_2$ in $P$, $x_1$ must be visited before
  $y_1$ in $P$. All nodes in $P$ in between $y_1$ and $e_2$ must be
  outside $A \cap B$, since the only way to re-enter $A \cap B$ is
  through $e_2$. Therefore, $y_2$ must be visited before $x_2$ in
  $P$. This implies that $x_2$ is visited after $x_1, y_1$, and $y_2$,
  and since $P$ is a path, none of them can be revisited. Therefore,
  all remaining nodes in $P$ must be in $A \cap B$, since there are no
  remaining edges in $\delta_F(A \cap B)$. This is a contradiction,
  since $t \notin A \cap B$.

  The case where $e_1, e_2 \in \delta_F(A \cup B)$ is analogous: $P$
  must traverse $e_1$ to leave $A \cup B$, and traverse $e_2$ to
  re-enter $A \cup B$, but $t \notin A \cup B$.

  Therefore, $A \cup B$, $A \cap B$ in $\calC$. Note that any edge
  crossing $A \cap B$ or $A \cup B$ must also cross $A$ or $B$. Since
  $\delta_{F \cap \calS}(A) \cup \delta_{F \cap \calS}(B) \subseteq
  \{e_1, e_2\}$,
  $\delta_{F \cap \calS}(A \cap B) \cup \delta_{F \cap \calS}(A \cup
  B) \subseteq \{e_1, e_2\}$ as well.
\end{proof}

Let $\calS' \subseteq F \cap \calS$ be the subset of safe edges cross at least one violated cut. Notice that Corollary~\ref{st_22_shared_edge} proves that all violated cuts
crossed by the same safe edge form a ring family, and the above
Lemma~\ref{st_22_path_uncross} proves that all violated cuts crossed
by safe edges on the same $s$-$t$ path form a ring family. Therefore,
it suffices to show that there are three $s$-$t$ paths whose union
covers all edges in $\calS'$.

\begin{lemma}
\label{st_22_flow}
There exist three $s$-$t$ paths $P_1, P_2, P_3$ in $(V, F)$ s.t.
$S' \subseteq \bigcup_{i=1}^3 P_i$. Furthermore, we can find these
paths in polynomial time.
\end{lemma}
\begin{proof}
  We can assume that $\calC$ is not empty and hence there is at least
  one violated cut.  Consider a flow network on the graph $(V, F)$,
  where each safe edge is given a capacity of 2, and each unsafe edge
  is given a capacity of 1. Note that all violated cuts have a total
  capacity of exactly 4, and since violated cuts separate $s$ from
  $t$, the maximum $s$-$t$ flow is at most 4. If the maximum flow was
  less than 4, then by the max-flow min-cut theorem, there would be an
  $s$-$t$ cut with total capacity strictly less than 4. However, we
  constructed $F$ ensuring that all $s$-$t$ cuts have capacity at
  least 4. Therefore, the maximum flow must be exactly 4.

  All capacities are integral, hence there exists some integral
  max-flow; let $g$ be such a flow. Since the violated cuts have total
  capacity exactly 4, all edges crossing them must be fully
  saturated. In particular, all safe edges $e \in S'$ must have
  $g(e) = 2$. By flow decomposition, there are four $s$-$t$ paths
  $P_1, \dots, P_4$ each carrying a flow of 1, where all $e \in E$ with $g(e) > 0$ are in at
  least one path. Since each safe edge $e \in \calS'$ has $g(e) = 2$,
  $e$ belongs to at least two of the paths $P_1,\ldots,P_4$. Thus,
  choosing any three of them would cover all edges in $e \in S'$,
  hence $\calS' \subseteq \cup_{i=1}^3 P_i$.
\end{proof}

\begin{claim}
There is a $5$-approximation for $(2,2)$-Flex-ST.
\end{claim}
\begin{proof}
  Let $P_i$ be the paths defined by Lemma~\ref{st_22_flow}, and let
  $\calC_i$ be the set of all violated cuts whose corresponding safe
  edge is on the path $P_i$, i.e.
  $\calC_i = \{A \in \calC: \delta_{F \cap \calS}(A) \subseteq
  P_i\}$. By definition of $\calS'$, for all $A \in \calC$, $\delta_{F \cap \calS}(A) \subseteq S' \subseteq \cup_{i=1}^3 P_i$. Therefore,
\begin{align*}
  \calC = = \{A \in \calC: \delta_{F \cap S}(A) \subseteq \cup_{i=1}^3 P_i\} = \bigcup_{i=1}^3 \calC_i.
\end{align*}

We can solve the augmentation problems of finding the minimum cost
subset $F_i \subseteq E \setminus F$ s.t. $\delta_{F_i}(A) \geq 1$ for
all $A \in \calC_i$ for each $i = 1,2,3$. Since
$\calC \subseteq \cup_{i=1}^3 \calC_i$, $F' = \cup_{i=1}^3 F_i$ is a
feasible solution to the augmentation problem of finding the minimum
cost subset of $E \setminus F$ s.t. $\delta_{F'}(A) \geq 1$
for all $A \in \calC$. Therefore, $F \cup F'$ is a feasible solution
to $(2,2)$-Flex-ST.

As an immediate corollary of Claim~\ref{st_22_path_uncross}, each
$\calC_i$ is a ring family. From the discussion in
Section~\ref{sec:prelim}, we see that the three corresponding
augmentation problems can be solved exactly, so
$\cost(F_i) \leq \cost(F^*)$.  Thus,
$\cost(F') \leq 3 \cdot \cost(F^*)$. By
Lemma~\ref{st_22_starting_criteria},
$\cost(F) \leq 2\cdot \cost(F^*)$. Overall,
$\cost(F \cup F') \leq 5 \cdot \cost(F^*)$.

Finally, note that we can find $F$ in polynomial time by
Lemma~\ref{st_22_starting_criteria}. For each $F_i$, we can find the
minimal violated set by finding all cuts in in $F$ with $p+q$ (in this case, 4) edges,
as explained in Section~\ref{sec:prelim}.
\end{proof}

\mypara{An approach to approximate $(p,q)$-Flex-Steiner:} In the
appendix (Section \ref{sec:steiner}) we outline a candidate approximation algorithm for
$(p,q)$-Flex-Steiner that leverages an approximation algorithm for
$(p,q)$-Flex-ST. We state a conjecture that can lead to a provable
approximation guarantee.

\section{Approximation Algorithms for FGC}
\label{sec:fgc}
In this section we prove Theorem~\ref{thm:introfgc}. The theorem encapsulates two results which share the high-level approach but differ in some of the
details. We first prove an $O(q)$ approximation for $(2,q)$-FGC and then prove our results for $(p,q)$-FGC for $q \le 4$.
As discussed in Section~\ref{sec:prelim}, most of the proofs in this
section rely on proving that certain families of violated cuts are
uncrossable. Suppose $A,B \in \calC$ and the sets properly intersect.
We say that $A,B$ uncross if $A-B, B-A \in \calC$ or $A \cap B, A \cup B \in \calC$. Otherwise we say that $A,B$ do not uncross.
Note that if $A \cup B = V$ and $A$ and $B$ are violated, then by
symmetry, $V - A$ and $V - B$ are violated as well. In this case,
however, $V - A = B - A$ and $V - B = A - B$, so once again, $A$ and
$B$ uncross. Therefore, when proving uncrossability of two properly
intersecting sets $A$ and $B$, we assume  without loss  of generality
that $(A \cup B) \neq V$.

\subsection{An $O(q)$-approximation for $(2, q)$-FGC}
We prove that the augmentation problem for increasing flex-connectivity from $(2,q)$ to $(2,q+1)$, for any $q \ge 0$ corresponds to covering an uncrossable function.

\begin{lemma}
The set of all violated cuts when augmenting from $(2, q)$-FGC to $(2, q+1)$-FGC is uncrossable.
\end{lemma}
\begin{proof}
Let $F \subseteq E$ such that $F$ is a feasible solution to $(2, q)$-FGC. Suppose $A$, $B$ are two violated cuts that properly intersect. Since $A$ is a violated cut, $\delta(A)$ must contain at most $1$ safe edge and at most $q+2$ total edges, similarly $\delta(B)$. Since they must satisfy the requirements of $(2, q)$-FGC, they must have exactly $q+2$ total edges.

First, we claim that if both $\delta_F(A \cap B)$ and
$\delta_F(A \cup B)$ have less than $2$ safe edges each then $A, B$
uncross. To see this, suppose $|\delta_{F \cap \calS}(A \cap B)| < 2$
and $|\delta_{F \cap \calS}(A \cup B)| < 2$. Since $F$ is feasible for
$(2,q)$-FGC, $|\delta_F(A \cap B)| \geq q+2$, and
$|\delta_F(A \cup B)| \geq q+2$. By submodularity of the cut function,
$|\delta_F(A \cap B)| + |\delta_F(A \cup B)| \leq |\delta_F(A)| +
|\delta_F(B)| = 2(q+2)$. Therefore,
$|\delta_F(A \cap B)| = |\delta_F(A \cup B)| = q+2$. Since neither of
them have two safe edges and both have exactly $q+2$ edges, they are
violated, and hence $A$ and $B$ uncross.

By the same reasoning, using posimodularity of the cut function, we
can argue that at least one of $\delta_F(A - B)$ and $\delta_F(B - A)$
must have at least $2$ safe edges otherwise both $A-B,B-A$ are
violated which implies that $A,B$ uncross. To see this, suppose
$|\delta_{F \cap \calS}(A - B)| < 2$ and
$|\delta_{F \cap \calS}(B - A)| < 2$. Since $F$ is feasible for
$(2,q)$-FGC, $|\delta_F(A - B)| \geq q+2$, and
$|\delta_F(B - A)| \geq q+2$. By posimodularity of the cut function,
$|\delta_F(A - B)| + |\delta_F(B - A)| \leq |\delta_F(A)| +
|\delta_F(B)| = 2(q+2)$. Therefore,
$|\delta_F(A - B)| = |\delta_F(B - A)| = q+2$. Since neither of them
have two safe edges and both have exactly $q+2$ edges, they are
violated.

Thus, we can assume that one of $\delta_F(A \cap B)$ and
$\delta_F(A \cup B)$ has $2$ safe edges and one of $\delta_F(A - B)$
and $\delta_F(B - A)$ has $2$ safe edges.  Without loss of generality
assume that $|\delta_{F \cap \calS}(A \cap B)| \geq 2$ and
$|\delta_{F \cap \calS}(A - B)| \geq 2$, the other cases are
similar. Recall that $\delta_F(A)$ and $\delta_F(B)$ each have at most
one safe edge. Therefore,
$|\delta_{F \cap \calS}(A \cap B) \cup \delta_{F \cap \calS}(A - B)|
\leq 2$, so there must be two safe edges in
$\delta_F(A - B) \cap \delta_F(A \cap B)$, i.e. both safe edges must
be shared between the boundaries of $A - B$ and $A \cap B$. However
this implies that $\delta_F(B)$ has two safe edges contradicting that
$B$ was violated.
\end{proof}

The preceding lemma yields an $2(q+1)$-approximation for $(2,q)$-FGC
as follows. We start with a $2$-approximation for $(2,0)$-FGC that can
be obtained by using an algorithm for $2$-ECSS. Then for we augment in
$q$-stages to go from a feasible solution to $(2,0)$-FGC to
$(2,q)$-FGC. The cost of augmentation in each stage is at most $\opt$
where $\opt$ is the cost of an optimum solution to $(2,q)$-FGC. We can
use the known $2$-approximation algorithm in each augmentation stage
since the family is uncrossable. Recall from Section~\ref{sec:prelim}
that the the violated cuts can be enumerated in polynomial time, and
hence the primal-dual $2$-approximation for covering an uncrossable
function can be implemented in polynomial-time. This leads to the
claimed approximation and running time.

\subsection{Approximating $(p, q)$-FGC for $q \leq 4$}
We have seen that the augmenting problem from $(2,q)$-FGC to
$(2,q+1)$-FGC leads to covering an uncrossable function. Boyd et
al. \cite{BoydCHI22} showed that augmenting from $(p,0)$-FGC to $(p,1)$-FGC
also leads to an uncrossable function for any $p \ge 1$. However this approach fails in
general for most cases of augmenting from $(p,q)$-FGC to
$(p,q+1)$-FGC. We give an example for $(3,1)$ to $(3,2)$ in
Figure~\ref{fgc_32_counterex} in Section~\ref{sec:counterexamples}.
However, in certain cases, we can argue that we can take a more
sophisticated approach where we solve the augmentation problem by
considering the violated cuts in a small number of stages: in each
stage we choose a subfamily of the violated cuts that is
uncrossable.  In such cases, we can get obtain a $2k$-approximation
for the augmentation problem, where $k$ is the upper bound on the
number of stages. Here we show that this approach works to augment
from $(p,q)$ to $(p,q+1)$ whenever $q \le 2$ and also for $q=3$ when
$p$ is even. We show in Section~\ref{sec:counterexamples} limitations
of our approach for $q \geq 4$ and for $q =3$ when $p$ is odd.

Suppose we want to augment from $(p, q)$-FGC to $(p, q+1)$-FGC. Let $G = (V, E)$ be the original input graph, and let $F$ be the set of edges we have already included. Note that in $(V, F)$, all cuts fall into one of the following categories:
\begin{itemize}
    \item $\geq p$ safe edges
    \item $\geq p + q + 1$ total edges
    \item Exactly $p+q$ total edges, with at most $p-1$ safe edges
\end{itemize}
Note that the first two satisfy the constraints of $(p, q+1)$-FGC, so
the last category is exactly the set of all violated cuts. Instead of
attempting to cover all violated sets at once, we do it in stages
where in each stage we consider the violated cuts based on the number of
safe edges. We begin by covering all violated sets with no safe edges,
then with one safe edge, and iterate until all violated sets are
covered. This is explained in Algorithm~\ref{spanning_general_algo}
below.

\begin{algorithm}[H]
\caption{Augmenting from $(p,q)$ to $(p,q+1)$ in stages}
\label{spanning_general_algo}
\begin{algorithmic}[1]
    \State $H \gets F$
    \For{$i = 0, \dots, p-1$}
        \State $\calC_i \gets \{S : \emptyset \neq S \subsetneq V, |\delta_{H}(S)| = p+q, |\delta_{H \cap \calS}| = i\}$
        \State $F_i \gets$ approximation algorithm to cover cuts in $\calC_i$
        \State $H \gets H \cup F_i$
    \EndFor\\
    \Return H
\end{algorithmic}
\end{algorithm}

The only unspecified part of the algorithm is to cover cuts in
$\calC_i$ in the $i$'th stage.  If we can prove that $\calC_i$ forms
an uncrossable family then we can obtain a $2$-approximation in each
stage. One could also come up with a more elaborate algorithm to cover
$\calC_i$. First, we prove a generic and useful lemma regarding cuts
in $\calC_i$.

For the remaining lemmas, we let $H_i \subseteq E$ denote the set of edges $H$ at the start of iteration $i$. In other words, $H_i$ is a set of edges such that for all $\emptyset \neq S \subsetneq V$, if $|\delta_{H_i}(S)| = p+q$, then $|\delta_{H_i \cap \calS}(S)| \geq i$.

\begin{lemma}
\label{spanning_case1_lemma}
Fix an iteration $i \in \{0, \dots, p-1\}$. Let $\calC_i$ be as defined in Algorithm~\ref{spanning_general_algo}. Then, if $A, B \in \calC_i$ and
\begin{enumerate}
    \item $|\delta_{H_i}(A \cap B)| = |\delta_{H_i}(A \cup B)| = p+q$, or
    \item $|\delta_{H_i}(A - B)| = |\delta_{H_i}(B - A)| = p+q$
\end{enumerate}
then $A$ and $B$ uncross, i.e. $A \cap B, A \cup B \in \calC_i$ or $A - B, B - A \in \calC_i$.
\end{lemma}
\begin{proof}
By submodularity and posimodularity of the cut function, $|\delta_{H_i}(A \cap B)| + |\delta_{H_i}(A \cup B)|$ and $|\delta_{H_i}(A - B)| + |\delta_{H_i}(B - A)|$ are both at most $2(p+q)$.
First, suppose $|\delta_{H_i}(A \cap B)| = |\delta_{H_i}(A \cup B)| = p+q$. By submodularity of the cut function applied to the safe edges $H_i \cap \cal S$,
\[
2i = |\delta_{H_i \cap \calS}(A)| + |\delta_{H_i \cap \calS}(B)| \geq |\delta_{H_i \cap \calS}(A \cup B)| + |\delta_{H_i \cap \calS}(A \cap B)|.
\]
Since we are in iteration $i$, neither $\delta_{H_i}(A \cup B)$ nor $\delta_{H_i}(A \cap B)$ can have less than $i$ safe edges. Therefore, $|\delta_{H_i \cap \calS}(A \cup B)| = |\delta_{F_i \cap \calS}(A \cap B)| = i$. Then, $A \cap B, A \cup B \in \calC_i$, so $A$ and $B$ uncross.
In the other case, where $|\delta_{H_i}(A - B)| = |\delta_{H_i}(B - A)| = p+q$, we can use posimodularity of the cut function to once again show that $|\delta_{H_i \cap \calS}(A - B)| = |\delta_{H_i \cap \calS}(B - A)| = i$, so they are both in $\calC_i$.
\end{proof}

Note that the preceding lemma holds for the high-level approach. Now we focus on
cases where we can prove that $\calC_i$ is uncrossable.

\begin{lemma}
\label{spanning-k3-lemma}
Fix an iteration $i \in \{0, \dots, p-1\}$. Let $\calC_i$ be as defined in Algorithm~\ref{spanning_general_algo}. Then, for $q \leq 2$, $\calC_i$ is uncrossable.
\end{lemma}
\begin{proof}
  Suppose $A, B \subseteq V$ such that $\delta_{H_i}(A)$ and
  $\delta_{H_i}(B)$ both have exactly $i$ safe and $p+q-i$ unsafe
  edges. Suppose for the sake of contradiction that they do not
  uncross. By Lemma~\ref{spanning_case1_lemma}, one of
  $\delta_{H_i}(A \cap B)$ and $\delta_{H_i}(A \cup B)$ must have at
  most $p+q-1$ edges, and the same holds for $\delta_{H_i}(A - B)$ and
  $\delta_{H_i}(B - A)$. Without loss of generality, suppose
  $\delta_{H_i}(A \cap B)$ and $\delta_{H_i}(A - B)$ each have at
  most $p+q-1$ edges. By the assumptions on $H_i$, they must both have
  at least $p$ safe edges, hence they each have at most $q-1$ unsafe
  edges. Note that
  $\delta_{H_i}(A) \subseteq \delta_{H_i}(A - B) \cup \delta_{H_i}(A
  \cap B)$, hence $\delta_{H_i}(A)$ can have at most $2(q-1)$ unsafe
  edges. When $q \leq 2$, $2(q-1) < q+1$, which implies that
  $\delta_{H_i}(A)$ has strictly more than $p-1$ safe edges, a
  contradiction. Notice that
  $\delta_{H_i}(A) \subseteq \delta_{H_i}(B - A) \cup \delta_{H_i}(A
  \cup B)$,
  $\delta_{H_i}(B) \subseteq \delta_{H_i}(A - B) \cup \delta_{H_i}(A
  \cup B)$, and
  $\delta_{H_i}(B) \subseteq \delta_{H_i}(B - A) \cup \delta_{H_i}(A
  \cap B)$; therefore the same argument follows regardless of which
  pair of sets each have strictly less than $p+q$ edges.
\end{proof}

\begin{corollary}
  For any $p \ge 2$ there is a $(2p+4)$-approximation for $(p,2)$-FGC
  and a $(4p+4)$-approximation for $(p,3)$-FGC.
\end{corollary}
\begin{proof}
  Via \cite{BoydCHI22} we have a $4$-approximation for $(p,1)$-FGC. We
  can start with a feasible solution $F$ for $(p,1)$-FGC and use
  Algorithm~\ref{spanning_general_algo} to augment from $(p,1)$ to
  $(p,2)$; each stage can be approximated to within a factor of $2$
  since $\calC_i$ is an uncrossable family. Since there are $p$ stages
  and the cost of each stage can be upper bounded by the cost of
  $F^*$, an optimum solution to the $(p,2)$ problem, the total cost of
  augmentation is $2p \opt$. This leads to the desired
  $(2p+4)$-approximation for $(p,2)$-FGC. For $(p,3)$-FGC we can augment
  from $(p,2)$ to $(p,3)$ paying an additional cost of $2p \opt$. This
  leads to the claimed $(4p+4)$-approximation. Following the discussion
  in Section~\ref{sec:prelim}, covering the uncrossable families that
  arise in Algorithm~\ref{spanning_general_algo} can be done in
  polynomial time.
\end{proof}

Can we extend the preceding lemma for $q = 3$? It turns out that it
does work when $p$ is even but fails for odd $p \ge 3$.

\begin{lemma}
  Fix an iteration $i \in \{0, \dots, p-1\}$. Let $\calC_i$ be as
  defined in Algorithm~\ref{spanning_general_algo}. Then, for $q = 3$
  and any \emph{even} integer $p$, $\calC_i$ is uncrossable.
\end{lemma}
\begin{proof}
  Suppose
  $A, B \subseteq V$ such that $\delta_{H_i}(A)$ and $\delta_{H_i}(B)$
  both have exactly $i$ safe and $p+3-i$ unsafe edges. Suppose for the
  sake of contradiction that they do not uncross. By
  Lemma~\ref{spanning_case1_lemma}, one of $\delta_{H_i}(A \cap B)$
  and $\delta_{H_i}(A \cup B)$ must have at most $p+2$ edges, and the
  same holds for $\delta_{H_i}(A - B)$ and $\delta_{H_i}(B -
  A)$. Without loss of generality, suppose $|\delta_{H_i}(A \cap B)|$
  and $|\delta_{H_i}(A - B)|$ each have at most $p+2$ edges. By the
  assumptions on $H_i$, they must both have at least $p$ safe
  edges. By submodularity of the cut function,
\[|\delta_{H_i \cap \calS}(A)| + |\delta_{H_i \cap \calS}(B)| = 2i \geq |\delta_{H_i \cap \calS}(A \cap B) + |\delta_{H_i \cap \calS}(A \cup B) \geq p + |\delta_{H_i \cap \calS}(A \cup B).\]
Similarly, by posimodularity of the cut function,
\[|\delta_{H_i \cap \calS}(A)| + |\delta_{H_i \cap \calS}(B)| = 2i \geq |\delta_{H_i \cap \calS}(A - B)| + |\delta_{H_i \cap \calS}(B - A)| \geq p + |\delta_{H_i \cap \calS}(B - A)|.\]
Therefore, $\delta_{H_i}(A \cup B)$ and $\delta_{H_i}(B - A)$ can each have at most $2i-p < i$ safe edges, so by the requirement on $H_i$, they must each have at least $p+4$ total edges. Once again applying submodularity of the cut function,
\[|\delta_{H_i}(A)| + |\delta_{H_i}(B)| = 2p + 6 \geq |\delta_{H_i}(A \cap B)| + |\delta_{H_i}(A \cup B)| \geq p + 4 + |\delta_{H_i}(A \cap B)|.\]
Similarly, by posimodularity of the cut function,
\[|\delta_{H_i}(A)| + |\delta_{H_i}(B)| = 2p+6 \geq |\delta_{H_i}(A -
  B)| + |\delta_{H_i}(B - A)| \geq p + 4+ |\delta_{H_i \cap \calS}(A -
  B)|.\]

Thus $\delta_{H_i}(A \cap B)$ and $\delta_{H_i}(A - B)$ can each have a total of at most $p+2$ edges.

Notice that each safe edge on $\delta_{H_i}(A \cap B)$ or
$\delta_{H_i}(A - B)$ must be in either $\delta_{H_i}(A)$ or
$\delta_{H_i}(A \cap B) \cap \delta_{H_i}(A - B)$. Let
$\ell = |\delta_{H_i \cap \calS}(A \cap B) \cap \delta_{H_i \cap
  \calS}(A - B)|$, i.e. the number of safe edges crossing both
$A \cap B$ and $A - B$. Since $|\delta_{H_i \cap \calS}(A)| = i$,
\[i = |\delta_{H_i \cap \calS}(A)| \geq |\delta_{H_i \cap \calS}(A \cap B)| + |\delta_{H_i \cap \calS}(A - B)| - 2\ell \geq 2p - 2\ell,\]
implying that $\ell \geq \frac{2p-i}2$. Note that $\delta_{H_i}(A) \subseteq \delta_{H_i}(A \cap B) \cup \delta_{H_i}(A - B)$. Furthermore, if an edge crosses both $A \cap B$ and $A - B$, then it must have one endpoint in each (since they are disjoint sets), and therefore does not cross $A$. Therefore,
\begin{align*}
    |\delta_{H_i}(A)| \leq |\delta_{H_i}(A \cap B)| + |\delta_{H_i}(A - B)| - 2\ell \leq 2(p+2) - 2 \cdot \left \lceil \frac{2p-i}{2} \right \rceil \leq 4+i
\end{align*}
When $i < p-1$, this is at most $p+2$. When $i = p-1$, since $p$ is
even, $2 \cdot \left \lceil \frac{2p-i}{2} \right \rceil = 2 \cdot
\left \lceil \frac{p+1}{2} \right \rceil = p+2$, and $2(p+2) - p+2 =
p+2$. In either case, we get $|\delta_{H_i}(A)| = p+2$, a
contradiction to the assumption on $A$.
\end{proof}

The preceding lemma leads to an $(6p+4)$-approximation for $(p, 4)$-FGC
when $p$ is even by augmenting from a feasible solution to
$(p,3)$, since we pay an additional cost of $2p\opt$. The reasoning in the preceding lemma also shows that
$\calC_i$ is uncrossable when $p$ is odd as well for all $i < p -
1$. Therefore, the bottleneck for odd $p$ is in covering
$\calC_{p-1}$.  We show, via an example that the family is indeed not
uncrossable (see Section~\ref{sec:counterexamples}).  However, it may
be possible to show that $\calC_{p-1}$ separates into a constant
number of uncrossable families leading to an $O(p)$-approximation for
$(p,4)$-FGC for all $p$. The first non-trivial case is when $p=3$.

\subsection{Examples}
\label{sec:counterexamples}
We describe several examples that demonstrate that violated cuts
that arise in augmentation do not form an uncrossable family.

\mypara{Augmenting from $(3,1)$-FGC to $(3,2)$-FGC:}
Consider the simple example in Figure~\ref{fgc_32_counterex}.
Let $A = \{x_1, x_2\}, B = \{x_2, x_3\}$. Notice that $A$ and $B$ are
violated, since they each have two safe and two unsafe edges, but
$A \cup B = \{x_1, x_2, x_3\}$ and $B - A = \{x_3\}$ are not, since
they each have three safe edges.

\begin{figure}
  \begin{center}
    \includegraphics[width = 0.3\linewidth]{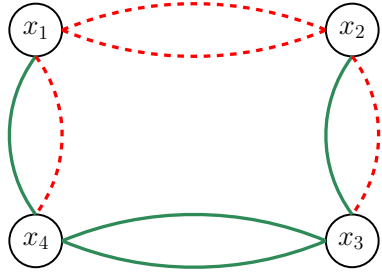}
  \end{center}
  \caption{Example graph where violated cuts in augmentation problem from $(3,1)$-FGC to $(3,2)$-FGC are not uncrossable. Red dashed edges are unsafe, while green solid edges are safe.}
  \label{fgc_32_counterex}
\end{figure}

\mypara{Augmenting from $(p,3)$-FGC to $(p,4)$-FGC when $p$ is odd:}
We show that the family of cuts with exactly $p-1$ safe edges
when augmenting from $(p, 3)$-FGC to $(p, 4)$-FGC is not uncrossable
when $p$ is odd. See Figure~\ref{fgc_k4odd_counterex}. Let
$A = \{x_1, x_2\}$, $B = \{x_2, x_3\}$. Notice that both $A$ and $B$
have exactly $p-1$ safe edges and $4$ unsafe edges. However,
$A \cup B$ and $B - A$ each have $p$ safe edges, and $A \cap B$ and
$A - B$ each have $p+4$ total edges. Note that unlike in the above
case, none of the cuts $A \cup B, A \cap B, A - B, B - A$ are
violated.

\begin{figure}
  \begin{center}
  \includegraphics[width = 0.4\linewidth]{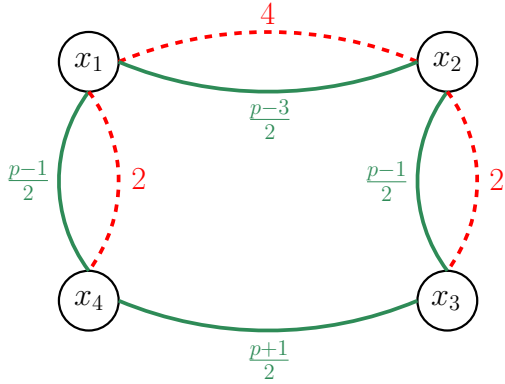}
  \end{center}
  \caption{Example graph where violated cuts in augmentation problem from $(p, 3)$-FGC to $(p, 4)$-FGC are not uncrossable when $p$ is odd. The numbers on edges denote the number of parallel edges. Red dashed edges are unsafe, while green solid edges are safe.}
  \label{fgc_k4odd_counterex}
\end{figure}

\mypara{Augmenting from $(4,4)$-FGC to $(4,5)$-FGC:} Finally, we
demonstrate an example where two violated cuts are not uncrossable
when augmenting from $(4, 4)$-FGC to $(4, 5)$-FGC, in
Figure~\ref{fgc_k5_counterex}. As above, let $A = \{x_1, x_2\}$, let
$B = \{x_2, x_3\}$. Note that $A$ and $B$ each only have 3 safe edges
crossing them, so the argument that $\calC_i$ is uncrossable for
$i < p-1$ fails. $A \cup B$ and $B - A$ each have at least $4$ safe
edges, and $A \cap B$ and $A - B$ each have $9$ total edges, so none
are violated.

\begin{figure}
  \begin{center}
    \includegraphics[width = 0.4\linewidth]{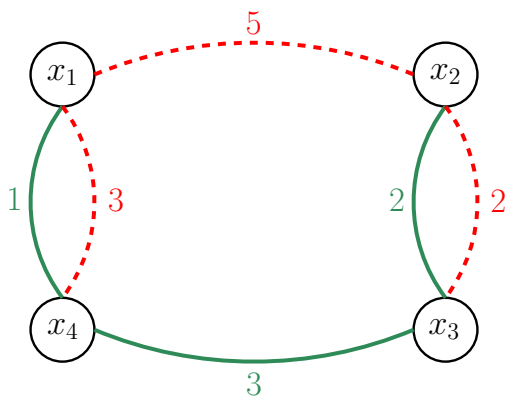}
  \end{center}
  \caption{Example graph where violated cuts in augmentation problem from $(4,4)$-FGC to $(4,5)$-FGC are not uncrossable. The numbers on edges denote the number of parallel edges. Red dashed edges are unsafe, while green solid edges are safe. }
  \label{fgc_k5_counterex}
\end{figure}

\bibliographystyle{plainurl}
\bibliography{fgc_paper}

\section{Appendix}
\label{sec:appendix_prelim_proofs}

\subsection{Proofs Omitted from Section~\ref{sec:prelim}}

\begin{proof}[Proof of Lemma~\ref{lemma:lp_polytime_fgc}]
  We show a polynomial time separation oracle for the given
  LP. Suppose we are given some vector $x \in [0,1]^{|E|}$. We first
  check if the capacitated min-cut constraints are satisfied in
  polynomial-time. This can be done in polynomial time by giving every
  safe edge a weight of $p+q$ and every unsafe edge a weight of $p$,
  and checking that the min-cut value is at least $p(p+q)$. If it is
  not, we can find the minimum cut and output the corresponding
  violated constraint. Suppose all capacitated constraints are
  satisfied. Now consider the first set of constraints.  If one of
  them is not satisfied there must must be some $S \subset V$ and some
  $B \subseteq \calU$, $|B| \le q$, such that
  $\sum_{e \in \delta(S) - B} x_e < p$. In particular,
  $\sum_{e \in \calS \cap \delta(S) - B} x_e < p$ and
  $\sum_{e \in \calU \cap \delta(S) - B} x_e < p$. We claim that the
  total weight (according to weights $(p+q)$ for safe edges and $p$
  for unsafe edges) going across $\delta(S)$ is at most $2p(p+q)$: at
  most $(p+q)p$ from $\calS \cap (\delta(S) - B)$, at most $p^2$ from
  $\calU \cap (\delta(S) - B)$, and at most $pq$ from $B$. Recall that
  the min-cut of the graph according the weights has already been
  verified to be at least $p(p+q)$. Hence, any violated cut from the
  first set of constraints corresponds to $2$-approximate min-cut.  It
  is known via Karger's theorem that there are at most $O(n^4)$
  $2$-approximate min-cuts in a graph, and moreover they can also be
  enumerated in polynomial time \cite{karger1993global,Karger00}. We
  can enumerate all $2$-approximate min-cuts and check each of them to
  see if they are violated. To verify whether a candidate cut $S$ is
  violated we consider the unsafe edges in $\delta(S) \cap \calU$ and
  sort them in decreasing order of $x_e$ value. Let $B'$ be a prefix
  of this sorted order of size $\min\{q, |\delta(S) \cap \calU|\}$. It
  is easy to see that that $\delta(S)$ is violated iff it is violated
  when $B = B'$. Thus, we can verify all candidate cuts efficiently.
\end{proof}

\begin{proof}[Proof of Lemma~\ref{lemma:augmentation_lp}]
  Suppose $S \subseteq V$ is violated in $(V,F_1)$, then
  $|\delta_{F_1}(S)| = p+q-1$, and $|\delta_{F_1 \cap \calS}(S)| <
  p$. Therefore, there are at least $q$ unsafe edges in
  $\delta_{F_1}(S)$. Let $B \subseteq \delta_{F_1}(S) \cap \calU$ such that $|B| =
  q$. Then, $|\delta_{F_1}(S) \setminus B| < p$, and since
  $x_e \leq 1$ for all $e \in E$,
  $\sum_{e \in \delta_{F_1}(S) - B} x_e < p$. By the first LP
  constraint, $\sum_{e \in \delta(S) - B} x_e \geq p$. Therefore,
  $\sum_{e \in \delta(S) \setminus F_1} x_e \geq 1$ as desired.
\end{proof}

\subsection{$\Omega(k)$ integrality gap for $(1,k)$-Flex-ST}
\label{sec:integrality_gap_1k}
For $(1,k)$-Flex-ST, Adjiashvili et al.\ \cite{AdjiashviliHMS20} showed an integrality gap
of $(k+1)$, in directed graphs, for an LP relaxation similar to the
one we described in this paper. They also showed a
poly-logarithmic factor inapproximability via a reduction from
directed Steiner tree. It is natural to ask whether the undirected
version of $(1,k)$-Flex-ST is super-constant factor hard when $k$ is large. As an
indication of potential hardness, we show an $\Omega(k)$-factor
integrality gap for the LP relaxation for $(1,k)$-Flex-ST. This is via
a simple modification of an example from \cite{ChakCKK15} that showed
an $\Omega(R)$-factor integrality gap for single-pair capacitated
network design, where $R$ is the connectivity requirement.

Consider a graph with $V = \{s, t\} \cup \{v_i: i \in [k+1]\}$, we add two parallel unsafe edges from $s$ to each $v_i$, and one safe edge from each $v_i$ to $t$. We give each unsafe edge a cost of $\frac 1 2$, and each safe edge a cost of $k+1$. See Figure ~\ref{1k_integrality_image}.

\begin{figure}
  \begin{center}
    \includegraphics[width = 0.4\linewidth]{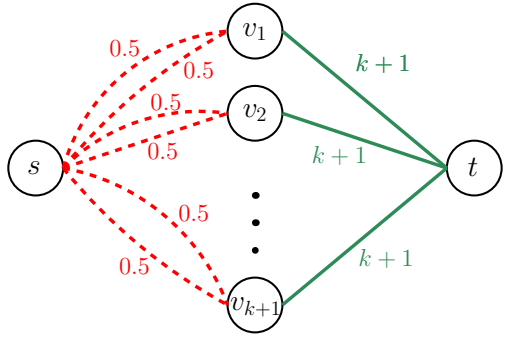}
  \end{center}
  \caption{Integrality gap for $(1,k)$-Flex-ST}
  \label{1k_integrality_image}
\end{figure}

\begin{claim}
Any optimal integral solution needs at least $\frac {k+1}2$ safe edges.
\end{claim}
\begin{proof}
Consider any integral solution $F \subseteq E$ where $|F \cap \calS| < \frac{k+1}2$. Let $S = \{s\} \cup \{v_i: \{v_i, t\} \notin F\}$. Then, $\delta_F(S)$ is exactly the set of all unsafe edges from $s$ to $v_i$ such that $\{v_i, t\} \in F$. Thus, $|\delta_F(S)| \leq 2 \cdot |F| < k+1$. Since $\delta_F(S)$ has no safe edges, $S$ is a violated cut, so $F$ is not a feasible solution.
\end{proof}

\begin{claim}
Consider the following fractional solution $x \in [0,1]^{|E|}$: $x_e = \begin{cases}
1 & e \in \calU \\
\frac{2}{k+1} &e\in \calS
\end{cases}$. Then, $x$ is a feasible LP solution.
\end{claim}
\begin{proof}
Let $S \subseteq V$ be any arbitrary cut such that $s \in S$, $t \notin S$. The capacitated constraint for $(1, k)$-flex connectivity is:
\[ (k+1) \sum_{e \in \delta_\calS(S)} x_e + \sum_{e \in \delta_\calU(S)} \geq k+1.\]
Substituting the values of $x_e$, we get
\[(k+1) \sum_{e \in \delta_\calS(S)} x_e + \sum_{e \in \delta_\calU(S)} = (k+1) \left(\frac 2 {k+1}\right)|\delta_\calS(S)| + |\delta_\calU(S)| = 2|\delta_\calS(S)| + |\delta_\calU(S)|.\]
Notice that $\delta(S)$ has one safe edge for each $v_i \in S$, since all safe edges are from $v_i$ to $t$ and $t \notin S$, so $|\delta_\calS(S)| = |S| - 1$. Similarly, $\delta(S)$ has two unsafe edges for every $v_i \notin S$, since $s \in S$. Thus $|\delta_\calU(S)| = 2(k+1 - (|S| - 1)) = 2(k - |S| + 2)$ Thus $2|\delta_\calS(S)| + |\delta_\calU(S)| = 2(|S| - 1) + 2(k - |S| + 2) = 2(k+1) \geq k + 1$.

For the remaining set of constraints, let $B \subseteq \calU$, $|B| = k$. We want to show that $\sum_{e \in \delta(S) - B} x_e \geq 1$. We case on $|S|$. If $|S| \geq \frac{k+3}2$, then by the analysis above, $|\delta_\calS(S)| = |S| - 1 \geq \frac{k+1}2$, so $\sum_{e \in \delta_\calS(S)} x_e \geq \frac{2}{k+1} \cdot \frac{k+1}{2} = 1$. Since $B \subseteq \calU$, $\sum_{e \in \delta(S) - B} x_e \geq 1$. Suppose instead that $|S| < \frac{k+3}2$. Then, since $x_e = 1$ for all unsafe edges,
\begin{align*}
    \sum_{\delta_\calU(S) - B} x_e \geq |\delta_\calU(S)| - |B| = 2(k - |S| + 2) - k = k + 4 - 2|S| > k+4 - (k+3) = 1
\end{align*}
In either case, we get our desired result.
\end{proof}

Notice that the total cost of any optimal integral solution is at least $(k+1) \frac {k+1}2 = \frac{(k+1)^2}2$. However, the optimal fractional solution has cost $k+1$ from the unsafe edges and $2(k+1)$ from the safe edges, for a total cost of $3(k+1)$. Thus, we get an $\Omega(k)$ integrality gap.

\subsection{An approach to approximate $(p,q)$-Flex-Steiner via $(p,q)$-Flex-ST}
\label{sec:steiner}

We consider the $(p,q)$-Flex-Steiner problem. The input is an
edge-weighted graph $G=(V,E)$ and a set $T \subseteq V$ of terminals
and our goal is to find a min-cost subset of edges $F \subseteq E$
such that in $(V,F)$ the terminals are $(p,q)$-flex-connected. We
first observed that $(p,q)$-flex-connectivity is symmetric and
transitive. Symmetry is easy to see since we are working in undirected
graphs. For transivity, if $u,v$ are $(p,q)$-flex-connected and $v,w$
are $(p,q)$-flex-connected then we claim that $u,w$ are
$(p,q)$-flex-connected. This too easily follows from the cut condition
for $(p,q)$-flex-connectivity. Thus, $(p,q)$-Flex-Steiner problem is
equivalent to the \emph{rooted} problem. In the rooted problem we are
given $G$, a root vertex $r \in V$ and set $T \subseteq V$ of
terminals and the goal is to find a min-cost subgraph $H=(V,F)$ such
that each $t \in T$ is $(p,q)$-flex-connected to $r$.

One can of course try to solve $(p,q)$-Flex-Steiner problem
directly. However, other than for $(1,q)$ and $(p,1)$ cases we do not
have any non-trivial results. It does not seem easy to generalize our
result for $(2,2)$-Flex-ST to $(2,2)$-Flex-Steiner
problem. Nevertheless, we can ask whether the single-pair algorithm
can somehow be used to develop a candidate algorithm for the Steiner
problem.  Suppose we have an $\alpha$-approximation for
$(p,q)$-Flex-ST.  Given a rooted $(p,q)$-Flex-Steiner instance we can
easily obtain an $ |T|\alpha$-approximation: for each $t \in T$ use
the approximation algorithm for $(p,q)$-Flex-ST to connect $t$ to $r$
and take the union of the solutions. Can we do better than this naive
approach?

Inspired by the success of a simple random permutation based greedy
algorithm for single-source vertex-connectivity problem
\cite{ChakrabortyCK08,ChuzhoyK08,ChekuriK08,Korula-thesis} we propose
the following algorithm which randomly permutes the terminals and
greedily connects the terminals to the root or previous terminals.

\begin{algorithm}[H]
\caption{Rooted $(p,q)$-Flex-Steiner via $(p,q)$-Flex-ST}
\label{steiner-algo}
\begin{algorithmic}[1]
    \State $F \gets \emptyset$
    \State Let $t_{i_1},t_{i_2},\ldots,t_{i_h}$ be a \emph{random} permutation of the terminal set $T \subset V$ where $h = |T|$
    \For{$j= 1, \dots, h$}
        \State $G_j$ is graph obtained by contracting $t_{i_1},\ldots,t_{i_{j-1}}$ into root $r$
        \State Use $(p,q)$-Flex-ST approx algorithm to connect $t_{i_j}$ to root $r$ in $G_j$. Let $F_j$ be the set of chosen edges.
        \State $F \gets F \cup F_j$
    \EndFor\\
    \Return $F$
\end{algorithmic}
\end{algorithm}

\begin{lemma}
Algorithm \ref{steiner-algo} outputs a feasible solution.
\end{lemma}
\begin{proof}
  Suppose for the sake of contradiction that some terminal is not
  $(p, q)$-flex-connected to the root $r$ in $(V, F)$, and let $j$ be
  the smallest index such that $t_{i_j}$ is not connected to the root
  $r$. Then, there must be some cut $S$ with $t_{i_j} \in S$,
  $r \notin S$, where $\delta_F(S)$ has $< p$ safe edges and $< p+q$
  total edges.

  First, suppose $t_{i_k} \in S$ for some $k < j$. Then, $S$ is an
  unsatisfied cut separating $t_{i_k}$ from $r$, implying that
  $t_{i_k}$ is not $(p, q)$-flex-connected to $r$. This contradicts
  the minimality of $j$.

  Therefore, $t_{i_k} \notin S$ for all $k < j$. Then, in the $j$th
  iteration of Algorithm~\ref{steiner-algo}, $G_j$ would not have
  contracted any edges in $\delta_F(S)$, and $S$ would have separated
  $t_{i_j}$ from the contracted root $r$, a contradiction.
\end{proof}

The analysis for the preceding random greedy algorithm for
vertex-connectivity is based on a cost sharing argument. We set up the
relevant notion and state a conjecture and its implication.  Given an
instance of rooted $(p,q)$-Flex-Steiner instance on $G=(V,E)$ with
root $r$ and terminal set $T$ we define the quantity $\beta_t$ for
each $t \in T$ as follows: $\beta_t$ is the minimum-cost to
$(p,q)$-flex-connect $t$ to $(T-t)+r$: in other words this is the cost
of $(p,q)$-flex-connecting $t$ to the root in the graph $G_t$ where we
contract all the terminal in $T-t$ into $r$.

\begin{conjecture}
\label{conj:cost-share}
There exists an integer valued function $f(p,q)$ such that
$\sum_{t \in T} \beta_t \le f(p,q) \opt$ where $\opt$ is the optimum
cost for the given rooted $(p,q)$-Flex-Steiner instance.
\end{conjecture}

Using a simple analysis that has been used previously
\cite{ChakrabortyCK08,ChuzhoyK08,ChekuriK08,Korula-thesis}, we obtain
the following.

\begin{lemma}
  If Conjecture~\ref{conj:cost-share} is true then
  Algorithm~\ref{steiner-algo} yields a randomized
  $O(f(p,q) \alpha(p,q) \ln |T|)$ approximation for rooted
  $(p,q)$-Flex-Steiner problem where $\alpha(p,q)$ is the
  approximation available for $(p,q)$-Flex-ST. Moroever there is also
  a deterministic variant of Algorithm~\ref{steiner-algo} that
  achieves the same approximation bound.
\end{lemma}
We omit the proof of the preceding lemma since it is an easy
consequence of the ideas in the previous work that we mentioned. We
believe that Conjecture~\ref{conj:cost-share} may be useful in
understanding the structure of $(p,q)$-Flex-Steiner problem.

\end{document}